%% file: arxiv_version.tex
\documentclass[11pt, letter]{article}

\usepackage[T1]{fontenc}    
\usepackage{hyperref}       
\usepackage{url}            
\usepackage{booktabs}       
\usepackage{amsfonts}       
\usepackage{nicefrac}       
\usepackage{microtype}      

\usepackage[compact]{titlesec}
\usepackage{enumitem}
\usepackage{wrapfig}

\usepackage{breakcites}

\usepackage{amsmath,amstext,amsopn,amsthm}
\usepackage{mathtools}
\usepackage{cleveref}
\usepackage{graphicx}
\usepackage{color}

\usepackage{bbm}
\usepackage{algorithm}

\usepackage[noend]{algpseudocode}
\newtheorem{theorem}{Theorem}
\newtheorem{corollary}{Corollary}

\newtheorem{proposition}{Proposition}

\makeatletter
\def\BState{\State\hskip-\ALG@thistlm}
\makeatother

\def\RR{{\mathbb R}}

\def\K{{\mathcal K}}

\def\M{{\mathcal M}}

\def\I{{\mathcal I}}

\DeclareMathOperator{\argmax}{argmax}

\DeclareMathOperator{\OPT}{OPT}
\DeclareMathOperator{\alg}{ALG}

\DeclareMathOperator{\ident}{\bf I}
\DeclareMathOperator{\SigmaB}{\bf \Sigma}
\DeclareMathOperator{\muB}{\bf \mu}
\DeclareMathOperator{\ub}{UB}
\DeclareMathOperator{\lb}{LB}

\usepackage{mathtools}
\renewcommand{\algorithmicrequire}{\textbf{Input:}}
\renewcommand{\algorithmicensure}{\textbf{Output:}}

\usepackage{sidecap}

\usepackage[suppress]{color-edits}
\addauthor{nh}{magenta}
\addauthor{at}{blue}

\usepackage[margin=1.1in]{geometry}

\usepackage{thmtools}
\usepackage{thm-restate}

\usepackage{hyperref} 
\usepackage[round]{natbib}
\title{Efficient algorithms for robust submodular maximization under matroid constraints}
\usepackage{times}

\author{Sebastian Pokutta\thanks{Georgia Institute of Technology,
    Atlanta. Email: sebastian.pokutta@isye.gatech.edu} \quad Mohit Singh\thanks{Georgia Institute of Technology, Atlanta. Email: mohitsinghr@gmail.com} \quad Alfredo Torrico\thanks{Georgia Institute of Technology, Atlanta. Email: atorrico3@gatech.edu}}

\date{}

\begin{document}

\maketitle

\begin{abstract}
In this work, we consider robust submodular maximization with matroid constraints.  We give an efficient bi-criteria approximation algorithm that outputs a small family of feasible sets whose union has (nearly) optimal objective value. This algorithm theoretically performs less function calls than previous works at cost of adding more elements to the final solution. We also provide significant implementation improvements showing that our algorithm outperforms the algorithms in the existing literature. We finally assess the performance of our contributions in three real-world applications.
\end{abstract}

\input{introduction.tex}

\input{related_work}

\input{contributions}

\input{experimental.tex}

\section*{Acknowledgements}
This research was partially supported by NSF Award CCF-NSF 1717947 and NSF CAREER award CMMI-1452463.
%

\bibliography{bibliography}

\appendix

\input{appendix.tex}

\end{document}

%% file: introduction.tex
\section{Introduction}
\label{sec:introduction}


In the last decade, submodular maximization \cite{fisher1978analysis,calinescu2011maximizing} has caught significant attention due to its applicability in numerous real-world applications, particularly those related to constrained subset selection problems, such as clustering \cite{gomes_krause10}, variable selection in graphical models \cite{krause_guestrin05} and sensor placement \cite{krause_etal08a,krause_etal09,powers_etal16}. A simple example is the problem of selecting a subset of patients that is the most informative in a group of people with certain illness. Submodularity reflects the decreasing marginal gain in the information acquired via bio-medical observations when choosing more patients \cite{krause_guestrin05}.
Formally, a set function $g:2^V\to\RR_+$ is \emph{submodular} if and only if \atdelete{it satisfies the {\it diminishing returns property}. Namely, }for any $e\in V$ and $A\subseteq B\subseteq V\backslash\{e\}$, $g_A(e)\geq g_B(e)$, where $g_A(e):=g(A+e)-g(A)$ and $A+e:= A\cup\{e\}$. Also, we say that $g$ is \emph{monotone} if for any $A\subseteq B\subseteq V$, we have $g(A)\leq g(B)$.
However, there are a few reasons that motivated constrained robust submodular optimization \cite{krause2008robust,anari2017robust}. For instance, when running medical tests on patients, any malfunction in the procedure will lead to imprecise observations, which translates in unstable models. Thus, the goal is to obtain solutions that are robust to these perturbations. This can be achieved by optimizing against (the minimum of) several submodular functions.

Our main contributions are: (a) to design an efficient bi-criteria algorithm that provides theoretical guarantees for robust submodular maximization subject to matroid constraints, (b) to observe that the main computational bottleneck in previous approaches is to \emph{certify} near-optimality of the obtained solution and present significant implementation improvements  to \atdelete{obtain short certificates of near-optimality}\atedit{attack this issue} and also (c) to show that our procedure can be \atdelete{easily implemented and} efficiently executed in real-world applications. 
Our theoretical results build on previous works related to \atdelete{fast and }efficient algorithms for submodular optimization and its robust variant~\cite{nemhauser1978analysis,minoux_78,krause2008robust,badanidiyuru_vondrak14,anari2017robust}.

\subsection{Problem Formulation}
Consider a collection of non-negative, monotone, submodular functions \(f_i: 2^V\to\RR_+\) on the same ground set \(V=\{1,\ldots,n\}\) with \(i\in[k]:=\{1,\ldots,k\}\), and also a family of feasible (also called independent) sets \(\I \), which define a matroid \(\M = (V,\I) \). Our interest is to obtain a feasible set \(S\in \I\) that maximizes the minimum over all objectives, i.e., we want to solve
\begin{equation}\label{eq:problem_def} \max_{S\in \I} \min_{i\in[k]} f_i(S) \end{equation}

\cite{krause2008robust} prove that problem \eqref{eq:problem_def} is NP-hard to approximate to any polynomial factor when any $k$ is considered. This \atdelete{hardness result }motivates the necessity of bi-criteria solutions. Specifically, in this work we focus on relaxing the constraints in order to get better approximation factors. We address this by constructing a small family of feasible sets whose union has (nearly) optimal objective value. This ``solution'' may initially seem counterintuitive for general matroids, however for many cases of interest, this is just a generalization of what is done when relaxing a single cardinality constraint. To exemplify this, consider \atdelete{the case of }partition constraints: here we are given a partition \(\{P_1, \ldots, P_q\}\) of the ground set and the goal is to pick a subset that includes at most \(b_j\) elements from part \(P_j\) for each \(j\). Then, the union of $\ell$ feasible sets have at most \(\ell\cdot b_j\) elements in each part. Since the output set $S$ is possibly infeasible, we define the \emph{violation ratio} $\nu$ as the minimum number of feasible sets whose union is $S$. In \atdelete{case of partition constraints}\atedit{our example}, this is equivalent to $\nu = \max_{j\in[q]}\lceil|S\cap P_j|/ b_j\rceil$.

%% file: related_work.tex
\subsection{Related Work}\label{sec:related_work}
There has been considerable work in robust submodular function maximization \cite{chekuri2010dependent,orlin_etal16,he_etal16,chen_etal16,chen2017robust,staib2017robust} and due to space limitations we will only be able to review work most closely related to ours. The initial model for robust submodular function maximization was introduced in \cite{krause2008selecting}. Later, \cite{krause2008robust} study the case when \( \I \) corresponds to cardinality constraints, and propose a (greedy-type) bi-criteria algorithm. \cite{powers2016constrained} considers the same robust problem with any matroid constraint but their relaxation approach is diferent. For structured combinatorial constraints, such as matroids or knapsack constraints, \cite{anari2017robust} propose an extended version of the standard greedy algorithm. 
However, in real-world applications the standard greedy algorithm is inefficient due to the number of function calls\atdelete{when the ground set $V$ is large}. Specifically, the extended greedy algorithm presented in \cite{anari2017robust} performs \(O(n r \ell)\) function evaluations, where \(r\) is the rank of the matroid and \(\ell \) is the number of rounds. Since \cite{fisher1978analysis}, there has been significant progress on reducing the number of evaluations, see e.g. \cite{minoux_78,badanidiyuru_vondrak14,mirzasoleiman_etal15} for vanilla version of submodular maximization. \atdelete{We build on these works to design efficient bi-criteria algorithms for problem \eqref{eq:problem_def}. }

%% file: contributions.tex
\subsection{Our Results and Contributions}\label{sec:contributions}
We propose an extended version of the \emph{threshold greedy} algorithm introduced by \cite{badanidiyuru_vondrak14}, which produces a family of feasible sets by using less function calls at cost of a small error in the approximation factor. The procedure is formally presented in Algorithm \ref{alg:ext_th_greedy} and its theoretical guarantee is given in Proposition \ref{prop:ext_threshold_greedy}.

\begin{algorithm}[h]
\caption{Extended Threshold-Greedy}\label{alg:ext_th_greedy}
\begin{algorithmic}[1]

\renewcommand{\algorithmicrequire}{\textbf{Input:}}
\renewcommand{\algorithmicensure}{\textbf{Output:}}
\Require $\ell\geq 1$, ground set $V$ with $n:=|V|$, monotone submodular function $g:2^{V} \rightarrow \RR_+$, matroid $\M=(V,\I)$ and $\delta>0$.
\Ensure feasible sets $S_1,\ldots, S_\ell\in \I$.
\For {$\tau=1, \dots, \ell$}
\State $S_\tau\leftarrow\emptyset$
\State $d\leftarrow\max_{e\in V} g(\cup_{j=1}^{\tau-1} S_j + e)$
\For {$(w=d; w\geq \frac{\delta}{n}d; w\leftarrow (1-\delta)w)$}
\For {$e\in V\backslash S_\tau$}
	\If {$S_\tau+e \in \I$ and $g_{\cup_{j=1}^{\tau} S_j}(e)\geq w$}
	\State $S_\tau\leftarrow S_\tau + e$
	\EndIf
\EndFor
\EndFor
\EndFor
\end{algorithmic}
\end{algorithm}
\normalsize

\begin{proposition}\label{prop:ext_threshold_greedy}
Given $\ell\geq 1$, a monotone submodular function $g:2^V\rightarrow \RR_+$ with $g(\emptyset)=0$, and parameter \(\delta >0\), Algorithm \ref{alg:ext_th_greedy} returns feasible sets $S_1,\ldots, S_\ell\in\I$ such that
\[g\left(\cup_{\tau=1}^\ell S_\tau\right)\geq \left(1-\left(\frac{1}{2-\delta}\right)^\ell \right)\cdot\max_{S\in \I} g(S).\]
This algorithm performs \(O( \frac{n\ell}{\delta}\log\frac{n}{\delta})\) function calls, \emph{independent} of the rank of the matroid.
\end{proposition}

By using Algorithm \ref{alg:ext_th_greedy}, we are able to construct a bi-criteria algorithm for problem \eqref{eq:problem_def} as follows: in a outer loop we obtain an estimate $\gamma$ on the value of the optimal solution $\OPT:=  \max_{S\in \I} \min_{i\in[k]} f_i(S)$ via a binary search. Next, for each guess $\gamma$ we define a new submodular set function \atdelete{$g^\gamma:2^{V}\rightarrow \RR_+$ }as \(g^\gamma(S):= \frac{1}{k} \sum_{i\in[k]} \min\{f_i(S), \gamma\}.\) Finally, given $\delta>0$, we run Algorithm \ref{prop:ext_threshold_greedy} (corresponding to a inner loop) on $g^\gamma$ with \(\ell =\lceil \log \frac{2k}{\epsilon}/\log (2-\delta)\rceil \) to obtain a candidate solution. Depending on this result, we update the binary search on $\gamma$, and we iterate. We stop the binary search whenever we get a relative error of $1-\epsilon/2$, namely, $(1-\epsilon/2)\OPT\leq\gamma\leq\OPT$. With this procedure we can get (nearly) optimal objective value by using less function evaluations than \cite{anari2017robust} at cost of producing a slightly bigger family of feasible sets. This is formally stated in Theorem \ref{theorem:offline}. 

\begin{theorem}\label{theorem:offline}
For problem \eqref{eq:problem_def}, there is a polynomial time algorithm that returns a set $S^{\alg}$, such that for given $0<\epsilon,\delta<1$, for all $j\in[k]$ it holds
\[f_j(S^{\alg})\geq (1-\epsilon) \cdot \max_{S\in \I} \min_{i\in[k]} f_i(S),\]
where $S^{\alg}=S_1\cup \dots\cup S_\ell$ with $\ell = \lceil \log \frac{2k}{\epsilon}/\log (2-\delta)\rceil$, and $S_1,\dots,S_\ell$ are feasible.
\end{theorem}

The proofs of these two theoretical results can be found in Appendix \ref{sec:proofs}. The extended threshold-greedy as well as other heuristics such as lazy evaluations~\cite{minoux_78} (see Appendix \ref{sec:other_improv}) improve the running time theoretically as well as on practical instances, see Section \ref{sec:comp-results}. Unfortunately, the main bottleneck remains obtaining a certificate of (near)-optimality or equivalently, a good upper bound on the optimum. We obtain that the optimum value is at most $\gamma$ whenever running the extended greedy algorithm on function $g^\gamma$ fails to return a solution of desired objective. Unfortunately, due to the desired accuracy in binary search and the number of steps in extended greedy, obtaining good upper bounds on the optimum is computationally prohibitive. 
We resolve this issue by implementing an early stopping rule in the bi-criteria algorithm. 
 When running Algorithm \ref{alg:ext_th_greedy} on  function \(g^\gamma\) (as explained above) we use the stronger guarantee given in Proposition~\ref{prop:ext_threshold_greedy}. When $\gamma$ is much larger than $\OPT$ and we fail to realize the guarantee in Proposition \ref{prop:ext_threshold_greedy}: if in iteration $\tau\in[\ell]$ we obtain a set $S_\tau$ such that $g(\cup_{t=1}^\tau S_t)< (1-1/(2-\delta)^\tau)\cdot\gamma$, then we stop and update the upper bound on the optimum to be $\gamma$. This allows us to stop the iteration much earlier since in many real instances $\tau$ is typically much smaller than $\ell$ when $\gamma$ is large. This leads to a drastic improvement in the number of function calls as well as CPU time. Indeed, without this improvement, the extended greedy algorithm of \cite{anari2017robust} even with lazy evaluations has a poor performance with a CPU time of more than 4 hours in small instances of $n=5,000$ elements as compared to few minutes after addition of this step. 



In Section \ref{sec:comp-results}, we assess the performance of these implementation improvements in two applications, showing empirically that the tested algorithms using these small changes significantly outperform the previous work. Also, we present a simple heuristic adapted from the \emph{stochastic greedy} algorithm introduced in \cite{mirzasoleiman_etal15}. We provide an extra experiment in Appendix \ref{sec:extra_experiment}.

%% file: experimental.tex
\section{Modeling and Experimental Results}
\label{sec:comp-results}
To facilitate the interpretation of our theoretical results, we will consider partition constraints in all experiments: the ground set $V$ is partitioned in $q$ sets $\{P_1,\ldots,P_q\}$ and the family of feasible sets is $\I = \{S: \ |S\cap P_j| \leq b, \ \forall j\in[q]\}$, same budget $b$ for each part. We test four methods: (\texttt{prevE-G}) the extended greedy with no improvements \cite{anari2017robust}, and the rest with improvements, (\texttt{E-G}) the extended greedy \cite{anari2017robust}, (\texttt{E-ThG}) the extended threshold greedy (this work), and (\texttt{E-StochG}) a heuristic we called extended stochastic greedy. The last procedure is an extended version of the \emph{stochastic greedy} \cite{mirzasoleiman_etal15}, and adapted to partition constraints (see Appendix \ref{sec:E-StochG}). Finally, we consider $\ell = \lceil \log\frac{2k}{\epsilon}\rceil$ for \texttt{E-G} and \texttt{E-StochG}, and $\ell = \lceil \log \frac{2k}{\epsilon}/\log (2-\delta)\rceil$ for \texttt{E-ThG}. See Appendix \ref{sec:final_pseudo} for the final pseudo-code of the main algorithm.

After running the four algorithms, we save the solution $S^{\alg}$ with the largest violation ratio $\nu$, and denote by $\tau_{\max} := \left\lceil \nu\right\rceil$. Observe that $S^{\alg} = S_1\cup\ldots\cup S_{\tau_{\max}}$ where $S_\tau\in\I$ for all $\tau\in[\tau_{\max}]$. We consider two additional baseline algorithms (without binary search): Random Selection (\texttt{RS}) which outputs a set $\tilde{S} = \tilde{S}_1\cup\ldots\cup \tilde{S}_{\tau_{\max}}$ such that for each $\tau\in[\tau_{\max}]$: $\tilde{S}_\tau$ is feasible, constructed by selecting elements uniformly at random, and $|\tilde{S}_\tau\cap P_j| = |S_\tau\cap P_j|$ for each part $j\in[q]$. Secondly,  (\texttt{G-Avg}) we run  $\tau_{\max}$ times the lazy greedy algorithm on the average function $\frac1k\sum_{i\in[k]}f_i$ and considering constraints $\I_\tau = \{S: \ |S\cap P_j| \leq |S_\tau\cap P_j|, \ \forall j\in[q]\}$ for each iteration $\tau\in[\tau_{\max}]$.

\atdelete{\subsection{Modeling and Experimental Results}}

In all experiments we consider the following parameters: approximation $1-\epsilon = 0.99$, threshold $\delta = 0.1$, and sampling in \texttt{E-StochG} with $\epsilon' = 0.1$.  \atdelete{The number of parts $q$ and budget $b$ for constraints $|S\cap P_j|\leq b$ will depend on the experiment.} The composition of each part $P_j$ is always uniformly at random from $V$.
%
%
%

\paragraph{Non-parametric Learning.}
We follow the setup in \cite{mirzasoleiman_etal15}. Let $X_V$ be a set of random variables corresponding to bio-medical measurements, indexed by a ground set of patients $V$. We assume $X_V$ to be a Gaussian Process (GP), i.e., for every subset $S\subseteq V$, $X_S$ is distributed according to a multivariate normal distribution $\mathcal{N}(\muB_S,\SigmaB_{S,S})$, where $\muB_S = (\mu_{e})_{e\in S}$ and $\SigmaB_{S,S} = [\K_{e,e'}]_{e,e'\in S}$ are the prior mean vector and prior covariance matrix, respectively. The covariance matrix is given in terms of a positive definite kernel $\K$, e.g., a common choice in practice is the squared exponential kernel $\K_{e,e'} =\exp(-\|x_e-x_{e'}\|^2_2/h)$. Most efficient approaches for making predictions in GPs rely on choosing a small subset of data points. For instance, in the Informative Vector Machine (IVM) the goal is to obtain a subset $A$ such that maximizes the information gain, \(f(A) = \frac{1}{2}\log\text{det}(\ident + \sigma^{-2}\SigmaB_{A,A})\). 
In our experiment, we use the Parkinson Telemonitoring dataset \cite{tsanas_etal10} consisting of $n = 5,875$ patients with early-stage Parkinsons disease and the corresponding bio-medical voice measurements with 22 attributes (dimension of the observations). We normalize the vectors to zero mean and unit norm. With these measurements we computed the covariance matrix $\Sigma$ considering the squared exponential kernel with parameter $h=0.75$. 
 For  our robust criteria, we consider $k=20$ perturbed versions of the information gain defined with $\sigma^2=1$, i.e., problem \eqref{eq:problem_def} corresponds to \(\max_{A\in \I}\min_{i\in[20]}f(A) + \sum_{e\in A\cap \Lambda_i}\eta_e\), where \(f(A)=\frac{1}{2}\log\text{det}(\ident + \SigmaB_{AA})\), \(\Lambda_i \) is a random set of size 1,000 with different composition for each $i\in[20]$, and $\eta\sim[0,1]^V$ is a uniform error vector.

We made 20 random runs considering $q = 3$ parts and budget $b=5$. We report the results in Figures \ref{fig:experiments} (a)-(d). In the performance profiles (a) and (b), we observe that any of the three algorithms clearly outperform \texttt{prevE-G}, either in terms of running time (a) or function calls (b). With this, we show empirically that our implementation improvements help in the performance of the algorithm. We also note that \texttt{E-StochG} is likely to have the best performance. Box-plots for the function calls in Figure \ref{fig:experiments} (d)  confirm this fact, since \texttt{E-StochG} has the lowest median. In this figure, we do not present the results of \texttt{prevE-G} because of the difference in magnitude of the number of function calls. Finally, in (c) we present the objective values obtained in a single run, and we observe that the stopping rule is useful since the three tested algorithms find a good solution earlier (using fewer elements) outperforming \texttt{prevE-G} and the benchmarks, and at much less cost as we mentioned before.

\begin{figure*}[ht]
\begin{center}
\begin{tabular}{ccc}
\def\arraystretch{0.5}
\includegraphics[width=0.3\textwidth]{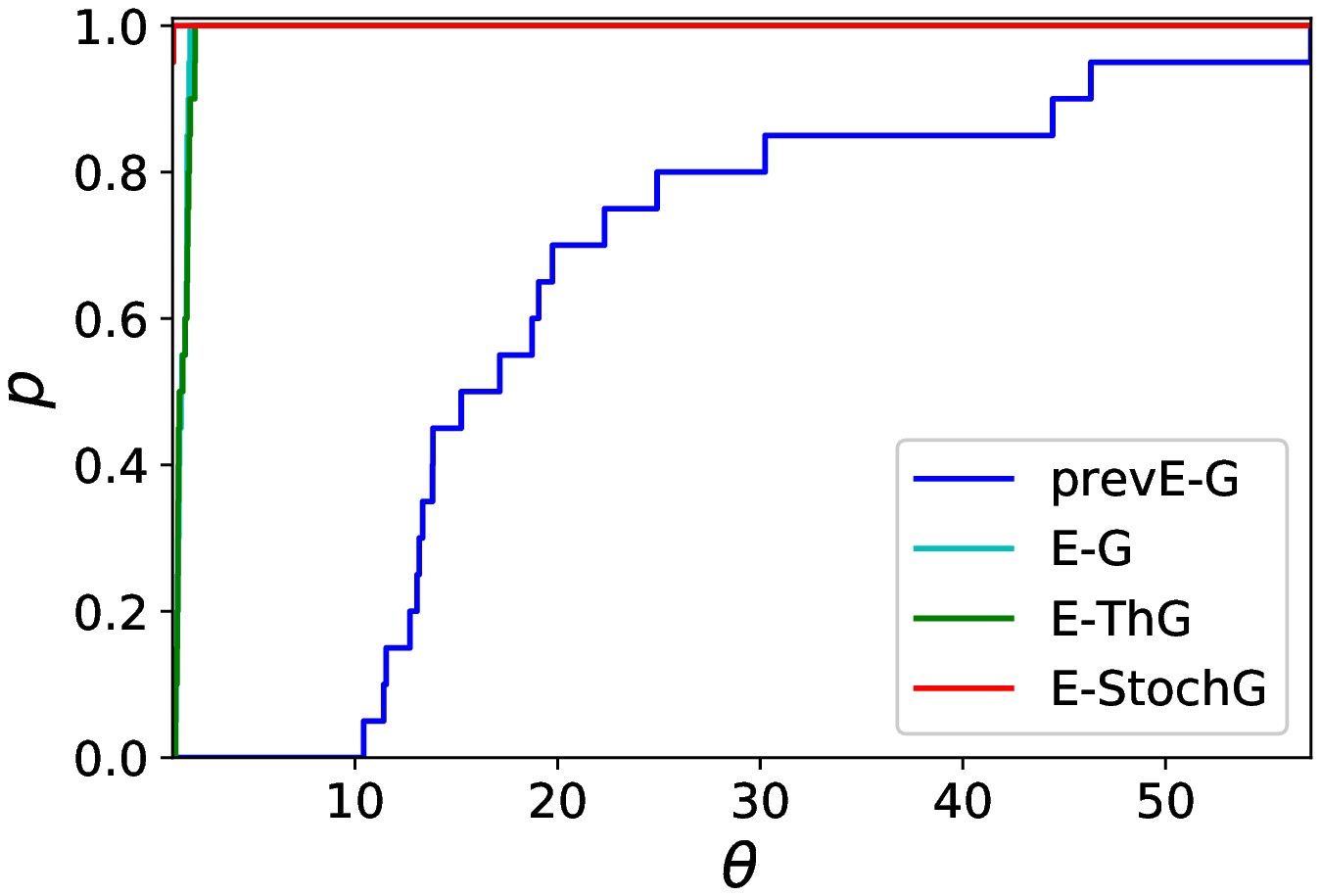} &
\includegraphics[width=0.3\textwidth]{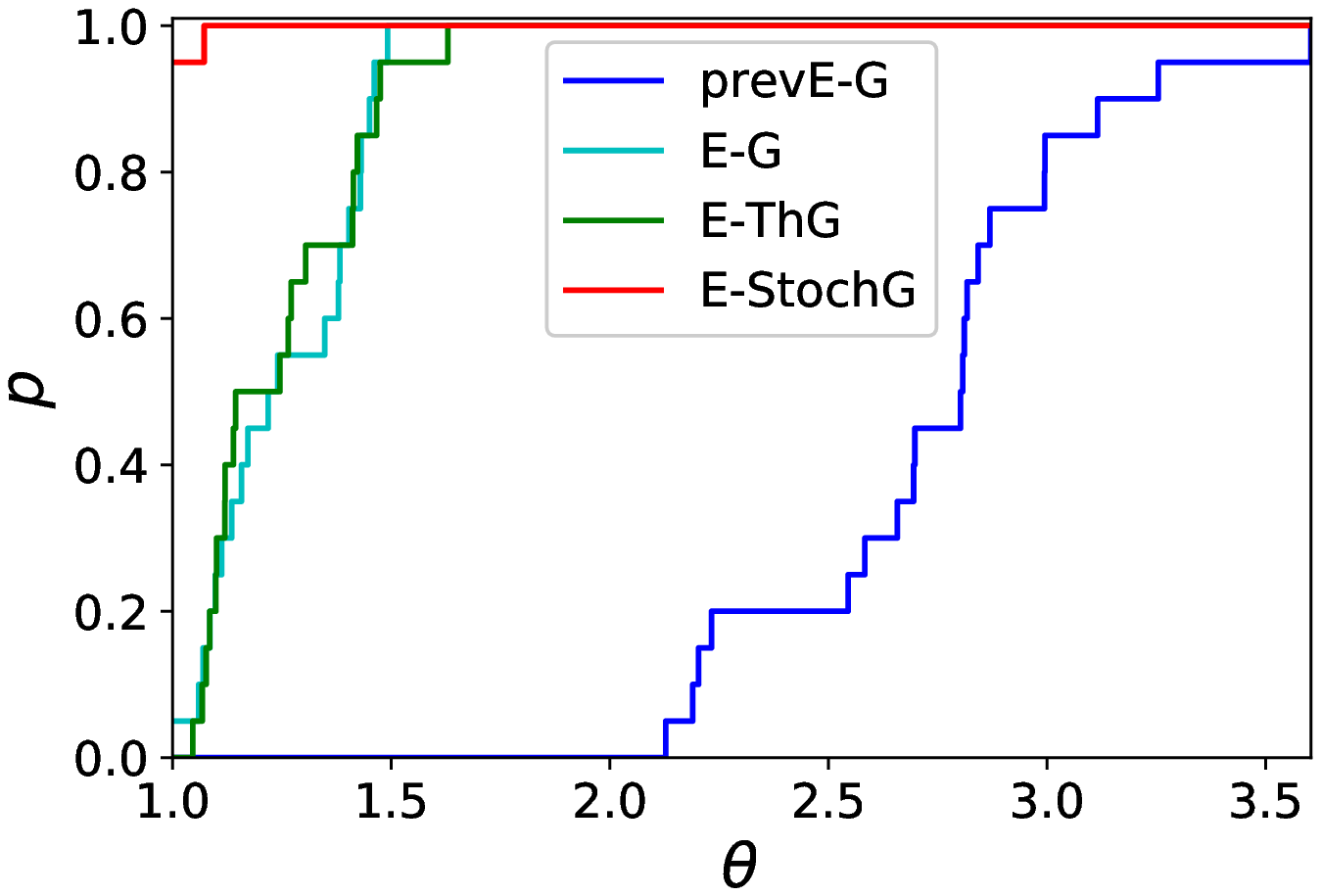} &
\includegraphics[width=0.3\textwidth]{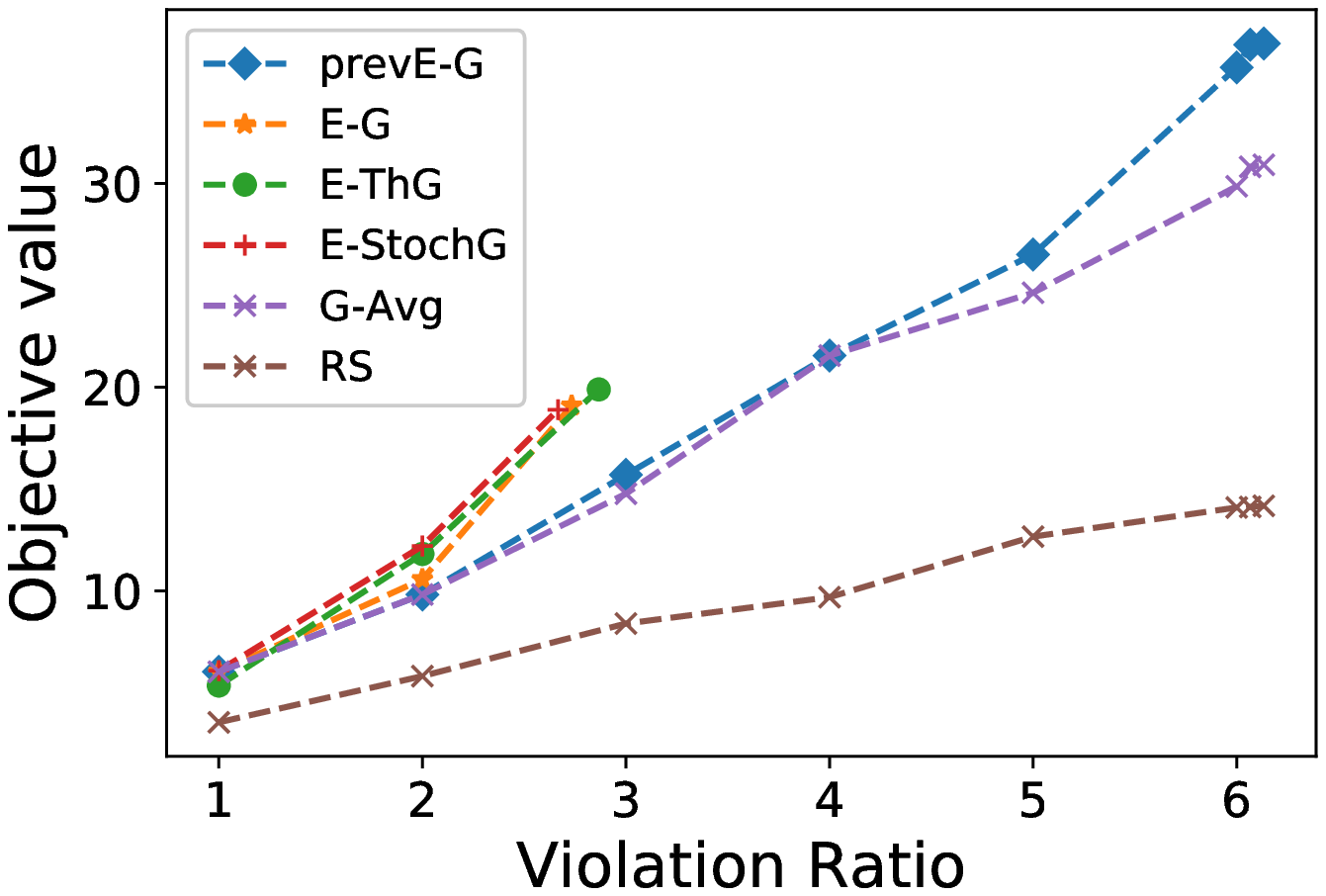} \\
 (a) & (b) & (c) \\
\includegraphics[width=0.3\textwidth]{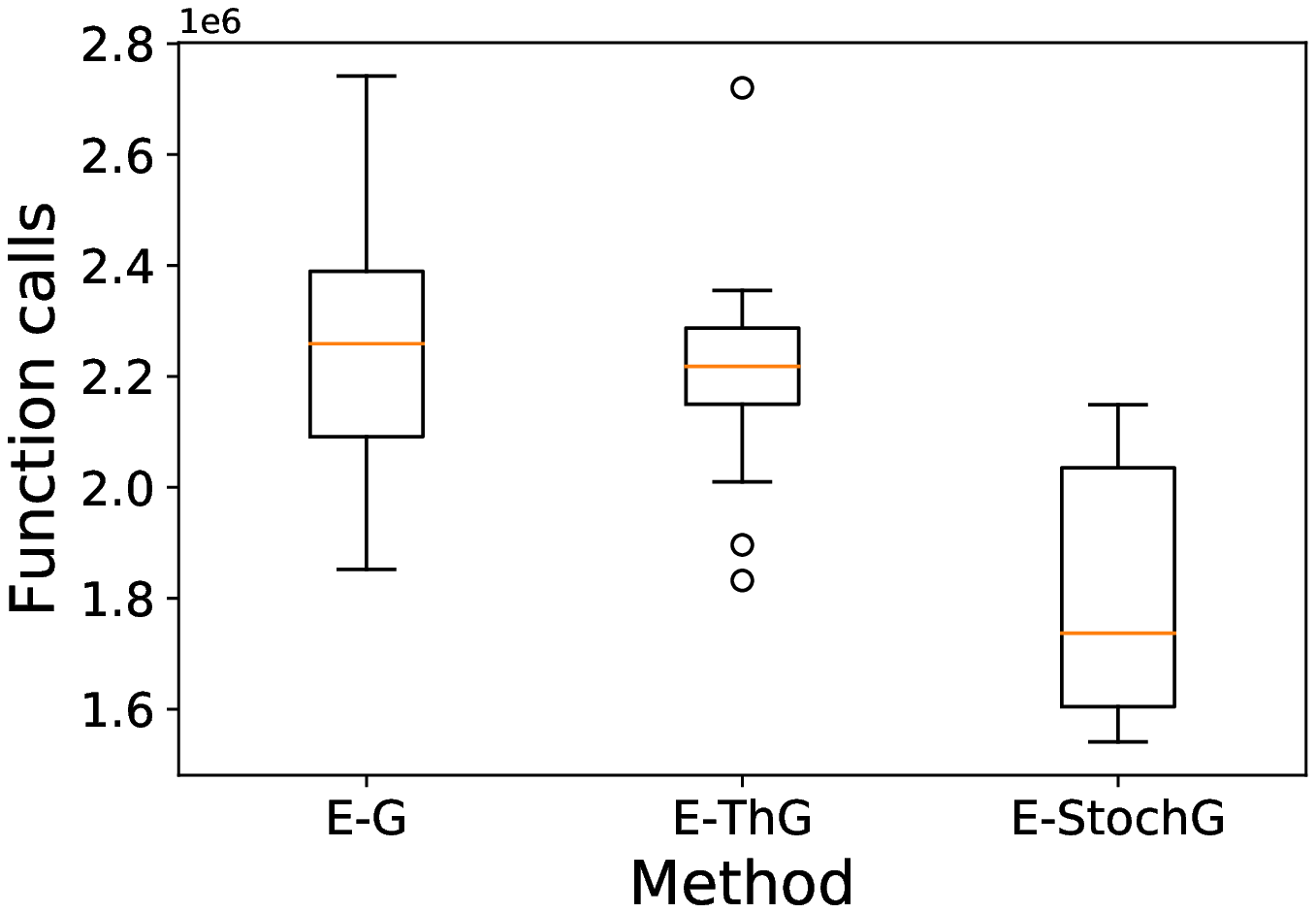} &
\includegraphics[width=0.3\textwidth]{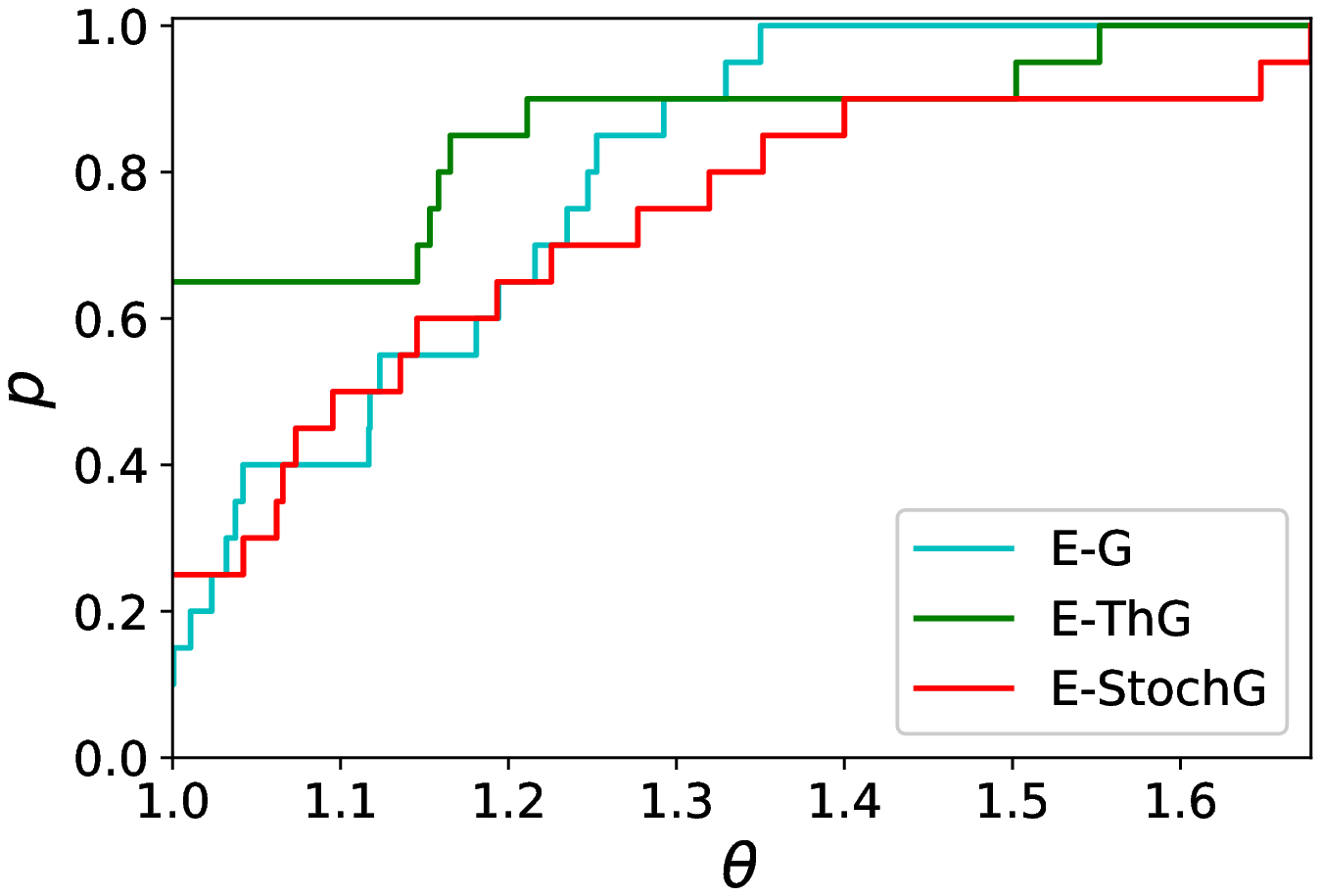}&
\includegraphics[width=0.3\textwidth]{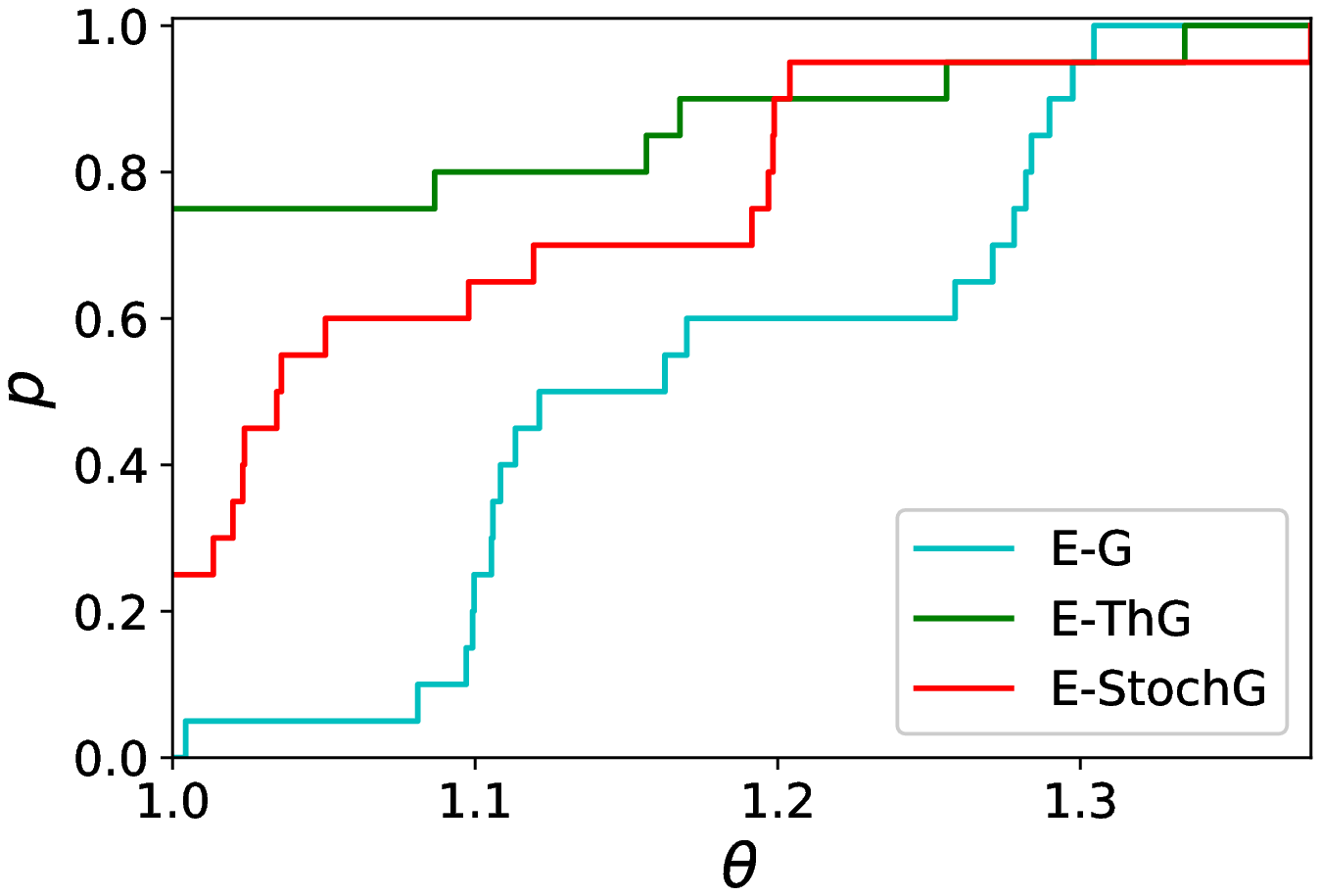}  \\
(d) & (e) & (f) \\
\includegraphics[width=0.3\textwidth]{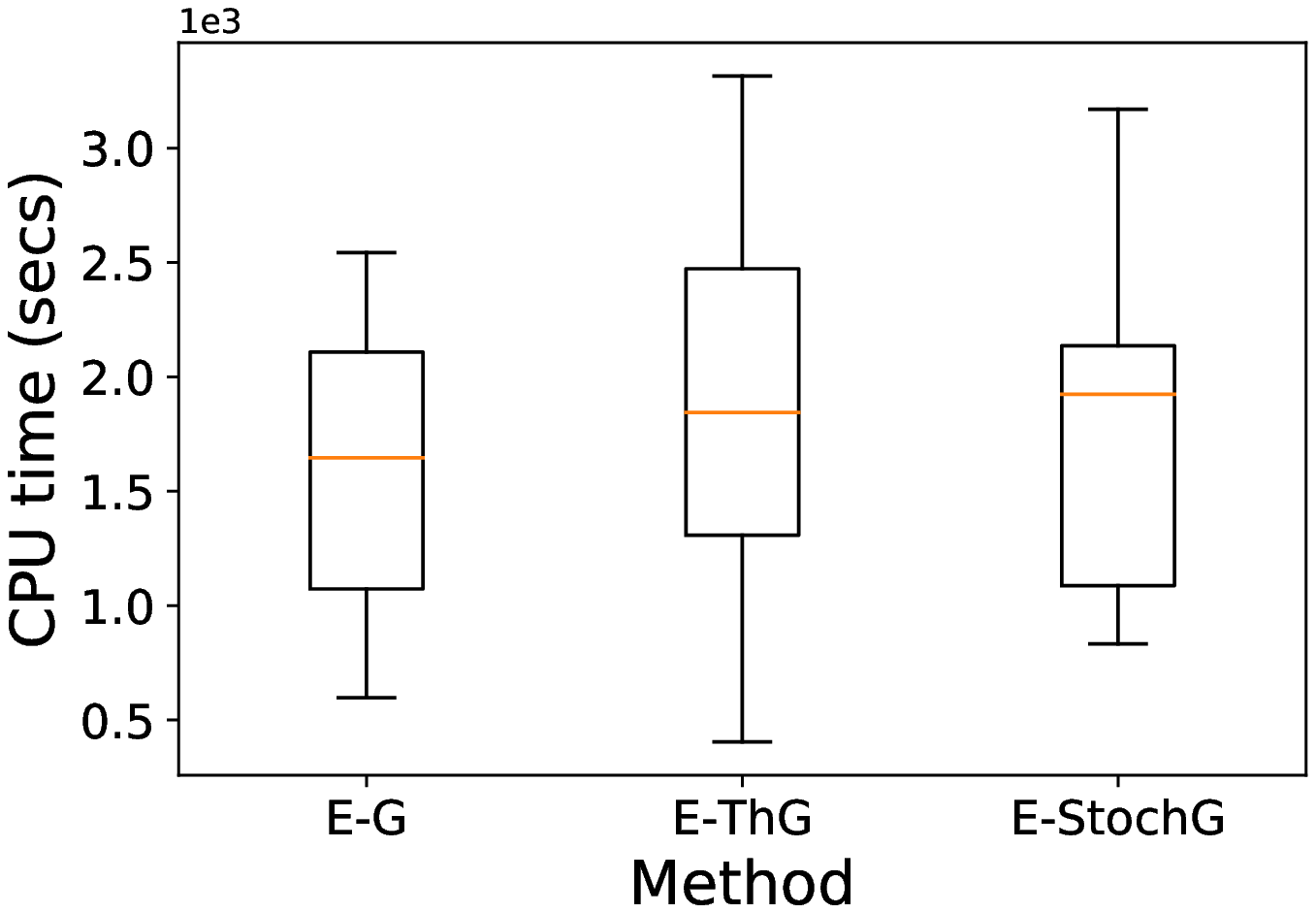} &
\includegraphics[width=0.3\textwidth]{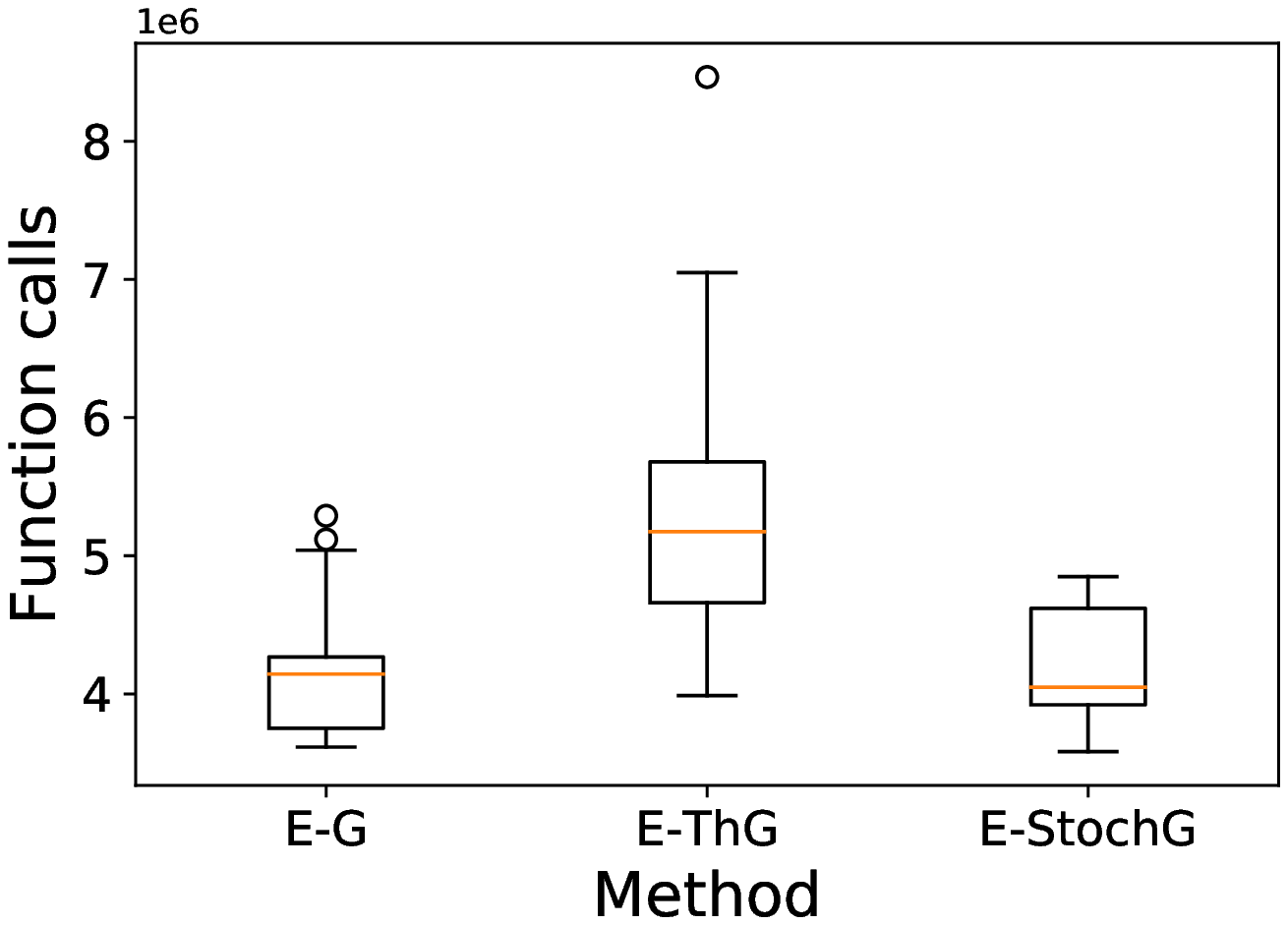} & \\
 (g) & (h) &
 \end{tabular}
\caption{{\it Non-parametric learning:} performance profiles (a) for running time (note that \texttt{E-G} is covered by \texttt{E-ThG}) and (b) for function calls. In (c) is the objective value versus the violation ratio in a single run of each method. In (d) is the box-plot for the function calls. {\it Clustering:} (\texttt{small}) performance profiles (e) for the running time and (f) for the function calls. (\texttt{Large}) box-plots (g) for the running time and (h) for the function calls.}
\label{fig:experiments}
\end{center}
\end{figure*}

\paragraph{Exemplar-based Clustering.}
We follow the setup in \cite{mirzasoleiman_etal15}. Solving the \(k\)-medoid problem is a common way to select a subset of exemplars that represent a large dataset $V$ \cite{kaufman_09}. This is done by minimizing the sum of pairwise dissimilarities between elements in $A\subseteq V$ and $V$. Formally, define \(L(A) =\frac{1}{V}\sum_{e\in V} \min_{v\in A}d(e,v)\), where \(d: V\times V\to\RR_+\) is a distance function that represents the dissimilarity between a pair of elements. By introducing an appropriate auxiliary element $e_0$, it is possible to define a new objective \(f(A) := L(\{e_0\}) - L(A+e_0)\) that is monotone and submodular \cite{gomes_krause10}, thus maximizing \(f\) is equivalent to minimizing $L$.
In our experiment, we use the VOC2012 dataset \cite{voc_2012}. The ground set $V$ corresponds to images, and we want to select a subset of the images that best represents the dataset. Each image has several (possible repeated) associated categories such as person, plane, etc. There are around 20 categories in total. Therefore, images are represented by feature vectors obtained by counting the number of elements that belong to each category, for example, if an image has 2 people and one plane, then its feature vector is $(2,1,0,\ldots,0)$ (where zeros correspond to other elements). We choose the Euclidean distance $d(e,e') = \|x_e-x_{e'}\|$ where $x_e,x_{e'}$ are the feature vectors for images $e,e'$. We normalize the feature vectors to mean zero and unit norm, and we choose $e_0$ as the origin. For our robust criteria, we consider $k=20$ perturbations of the function $f$ defined above, i.e., problem \eqref{eq:problem_def} corresponds to \(\max_{A\in \I}\min_{i\in[20]}f(A) + \sum_{e\in A\cap \Lambda_i}\eta_e\), where \(\Lambda_i \) is a random set of fixed size with different composition for each $i\in[20]$, and finally, $\eta\sim[0,1]^V$ is a uniform error vector.

We consider two experiments: (\texttt{small}) with $n=3,000$ images, 20 random instances considering $q=6$ and $b=70$, $|\Lambda_i|=500$ and (\texttt{large})  with $n=17,125$ images, 20 random instances $q\in\{10,\ldots,29\}$ parts and budget $b=5$, $|\Lambda_i|=3,000$ (we do not implement \texttt{prevE-G} because of the exorbitant running time). We report the results of the experiments in Figures \ref{fig:experiments} (e)-(h). For \texttt{small}, charts (e) and (f) confirm our theoretical results: \texttt{E-ThG} is the most likely to use less function calls (f) and running time (e) when the rank is relatively high (in this case $q\cdot b =420$) which contrasts with the performance of \texttt{E-G} that depends on the rank (chart (f) reflects this). For \texttt{large}, we can see in charts (g) and (h) that the results are similar, either in terms of running time or function evaluations, so when we face large ground sets, we could choose any algorithm, but we would still prefer \texttt{E-ThG} since it has no dependency on the rank.

%% file: appendix.tex
\section{Appendix}\label{sec:proofs}

\subsection{Additional proofs}
For the problem of maximizing a single submodular function subject to cardinality constraints, a variation of the standard greedy algorithm called {\it threshold greedy} (Algorithm \ref{alg:threshold_greedy} with $\I$  being a cardinality constraint) is proposed by \cite{badanidiyuru_vondrak14}. 
This algorithm achieves a \((1-1/e-\delta)\)-approximation factor, where \(\delta\) is the parameter for decreasing the threshold (line 4 in Algorithm \ref{alg:threshold_greedy}). Moreover, it can be easily adapted to any matroid constraint (Algorithm \ref{alg:threshold_greedy}), and it achieves a \((\frac{1-\delta}{2-\delta})\)-approximation factor. We formally state this result in Corollary \ref{cor:threshold_greedy}.

\begin{algorithm}[tb]
\caption{\small Threshold-Greedy for General Matroid Constraints}\label{alg:threshold_greedy}
\begin{algorithmic}
\small
\Require $g:2^V\to\RR_+$ monotone submodular, matroid $\M=(V,\I)$ and $\delta>0$.
\Ensure a set $S\subseteq V$, such that $S\in\I$.
\State \(S\leftarrow\emptyset\)
\State $d\leftarrow\max_{e\in V} g(e)$
\For {$(w=d; w\geq \frac{\delta}{n}d; w\leftarrow (1-\delta)w)$}
\For {$e\in V\backslash S$}
	\If {$S+e \in \I$ and $g_S(e)\geq w$} 
	\State $S\leftarrow S + e$
	\EndIf
\EndFor
\EndFor
\end{algorithmic}
\end{algorithm}

\normalsize

\begin{corollary}\label{cor:threshold_greedy}
Given $\delta>0$, Algorithm \ref{alg:threshold_greedy} gives a $\left(\frac{1-\delta}{2-\delta} \right) $-approximation for the problem of maximizing a single nonnegative, monotone, submodular function $g$ subject to a matroid constraint, using $O(\frac{n}{\delta}\log\frac{n}{\delta})$ queries.
\end{corollary}

\begin{proof}
Denote by $r$ the rank of matroid $\M$. Let $S^*=\{e_1^*,\ldots, e_r^*\}$ and  $S=\{e_1,\ldots,e_r\}$ be the optimal set and the set obtained with the algorithm, respectively. W.l.o.g, we can assume that $S^*$ and $S_G$ are both basis in $\M$, so there exists a bijection $\phi$ such that $\phi(e_i) = e_i^*$ for all $i\in[r]$. Denote by $S_{i-1}=\{e_1,\ldots,e_{i-1}\}$ the set of elements after iteration $i-1$. Observe that if  $e_i$ is the next element chosen by the algorithm and the current threshold value is $w$, then we get the inequalities
\[g_{S_{i-1}}(x) = \left\{\begin{matrix} \geq  w & \text{if} \ x=e_i \\ \leq w/(1-\delta) & \text{if} \ x\in V \ \text{s.t.}  \ S_{i-1}+x \in \I\end{matrix}\right.\]
This imples that $g_{S_{i-1}}(e_i)\geq (1-\delta)f_{S_{i-1}}(x)$ for all $x\in V\backslash S_{i-1}$ such that $S_{i-1} + x\in \I$. In particular for $x = e_i^*$ we have then
\[(1-\delta)g_{S_{i-1}}(e_i^*)\leq g(S_i) - g(S_{i-1}).\]
 On the other hand, if we apply submodularity twice we get
\[g(S^*)-g(S)\leq \sum_{i=1}^Kg_{S}(e_i^*)\leq \sum_{i=1}^Kg_{S_{i-1}}(e_i^*)\]
Using the previous two inequalities we get
\[(1-\delta)[g(S^*)-g(S)] \leq  \sum_{i=1}^K  g(S_i) - g(S_{i-1}) = g(S).\]
So we finally obtained
\[g(S)\geq \left(\frac{1-\delta}{2-\delta} \right) \cdot g(S^*)\]
as claimed.
\end{proof}

Following the same idea as in \cite{anari2017robust}, we then reuse the threshold greedy algorithm in an iterative scheme to construct a family of feasible sets whose union has an (nearly) optimal objective value, 
as stated in Proposition \ref{prop:ext_threshold_greedy}.

\begin{proof}[Proof of Proposition \ref{prop:ext_threshold_greedy}]
We use Corollary \ref{cor:threshold_greedy} to state that the extended threshold greedy algorithm when run for a single iteration returns a set $S_1\in \I$ such that
\[g(S_1)-g(\emptyset)\geq \left(1-\left(\frac{1}{2-\delta}\right)\right) \max_{S\in \I}  \left\{g(S)-g(\emptyset)\right\}.\]
We use the above statement to prove our theorem by induction. For $\tau=1$, the claim follows directly. Consider any $\ell\geq 2$. Observe that the algorithm in iteration $\tau=\ell$, is exactly the Threshold-Greedy algorithm run on submodular function $g':2^{V} \rightarrow \RR_+$ where $g'(S):= g(S\bigcup \cup_{\tau=1}^{\ell-1} S_\tau)$. This procedure returns $S_\ell$ such that
\[g'(S_\ell)- g'(\emptyset) \geq \left(1-\frac{1}{2-\delta}\right) \max_{S\in \I} \left(g'(S)-g'(\emptyset)\right)\]
which implies that
\begin{align*}
g\left(\cup_{\tau=1}^\ell S_\tau\right)&-g\left(\cup_{\tau=1}^{\ell-1}S_\tau\right)\geq \left(1-\frac{1}{2-\delta}\right) \left(\max_{S\in \I}g(S)-  g\left(\cup_{\tau=1}^{\ell-1}S_\tau\right)\right).
\end{align*}
By induction we know \(g\left(\cup_{\tau=1}^{\ell-1}S_\tau\right)\geq \left(1-\left(\frac{1}{2-\delta}\right)^{\ell-1}\right)\max_{S\in \I} g(S).\) Thus we obtain
\begin{eqnarray*}
 g\left(\cup_{\tau=1}^\ell S_\tau\right)&\geq& \left(1-\frac{1}{2-\delta}\right) \max_{S\in \I}g(S) +\left(\frac{1}{2-\delta}\right) g\left(\cup_{\tau=1}^{\ell-1}S_\tau\right)\\
 &\geq& \left(1-\left(\frac{1}{2-\delta}\right)^{\ell}\right)\max_{S\in \I} g(S)
\end{eqnarray*}
as claimed.
\end{proof}

 Note that the number of function calls \(O( \frac{n\ell}{\delta}\log\frac{n}{\delta})\) does not depend on the rank of the matroid $r$ as the algorithm proposed in \cite{anari2017robust}, which requires \(O(n\ell r)\) function evaluations. 

\subsection{Proof of Theorem \ref{theorem:offline}}\label{sec:bi-criteria_alg}
The formal proof of Theorem \ref{theorem:offline} follows the same lines as in \cite{anari2017robust}.

\begin{proof}[Proof of Theorem \ref{theorem:offline}]
Consider the family of monotone submodular functions \(\{f_i\}_{i \in [k]}\), $g^\gamma$ defined as above and parameter $\gamma$ with relative error of \(1-\frac{\epsilon}{2}\) . If we run the extended threshold greedy algorithm \ref{alg:ext_th_greedy} on $g^\gamma$ with $\ell\geq \lceil \log \frac{2k}{\epsilon}/\log(2-\delta) \rceil$, we get a set $S^{\alg}= S_1 \cup\cdots \cup S_\ell$, where $S_j\in\I$ for all $j\in[\ell]$. Moreover, Proposition~\ref{prop:ext_threshold_greedy} implies that
\(g^\gamma(S^{\alg})\geq \left(1-\frac{\epsilon}{2k}\right)\gamma.\)
Then, by a contradiction argument we can prove that
$f_i(S^{\alg})\geq \left(1-\frac{\epsilon}{2}\right)\gamma \geq (1-\epsilon)\OPT$,
for all $i\in[k]$ as claimed.
\end{proof}

\subsection{Other implementation improvements}\label{sec:other_improv}

\paragraph{Lazy evaluations} All algorithms and baselines are implemented with lazy evaluations \cite{minoux_78}. This means, we keep a list of an upper bound $\rho(e)$ on the marginal gain for each element (initially $\infty$) in decreasing order, and at each iteration, it evaluates the element at the top of the list $e'$. If the marginal gain of this element satisfies $g_S(e')\geq \rho(e)$ for all $e\neq e'$, then submodularity ensures $g_S(e')\geq g_S(e)$. In this way, greedy does not have to evaluate all marginal values to select the best element.

\paragraph{Bounds initialization.} To compute the initial $\lb$ and $\ub$ for the binary search, we run the lazy greedy \cite{minoux_78} for each function in a small sub-collection $\{f_i\}_{i\in[k']}$, where $k'\ll k$, leading to $k'$ solutions $A^1,\ldots, A^{k'}$ with guarantees $f_i(A^i)\geq (1/2)\cdot \max_{S\in \I}f_i(S)$. Therefore, we set $\ub = 2\cdot \min_{i\in[k']}f_i(A^i)$  and $\lb = \max_{j\in[k']}\min_{i\in[k]} f_i(A^j)$. This two values correspond to upper and lower bounds for the true optimum $\OPT$.

\begin{algorithm}[H]
\caption{\small Extended Stochastic-Greedy for Partition Matroid}\label{alg:ext_partition_stochastic_greedy}
\begin{algorithmic}[1]
\small
\Require $\ell\geq 1$, monotone submodular function $g:2^{V} \rightarrow \RR_+$, partition matroid $\M=(V,\I)$, $\epsilon'>0$.
\Ensure sets $S_1,\ldots, S_\ell\in \I$.
\For {$\tau=1, \dots, \ell$}
\State $S_\tau\leftarrow\emptyset$
\While {$S_\tau$ is not basis in $\M$}
\State For each $j\in[q]$, uniformly sample $R_j\sim P_j\backslash S_\tau$ with $\frac{n_j}{k_j}\log \frac{1}{\epsilon'}$ elements.
\State $e^*\leftarrow \argmax_{e\in R_1\cup\cdots R_r} \left\{g_{\cup_{j=1}^{\tau} S_j}( e)\right\}$.
\State $S_\tau\leftarrow S_\tau + e^*$.
\EndWhile
\EndFor
\end{algorithmic}
\end{algorithm}

\normalsize
\subsection{Extended Stochastic Greedy for Partition Matroid}\label{sec:E-StochG}
Consider a partition \(\{P_1,\ldots,P_q\}\) on ground set $V$ with $n_j := |P_j|$ for all $j\in[q]$ and a family of feasible sets $\I = \{S\subseteq V: \ |S\cap P_j|\leq k_j \ \forall j\in[q]\}$ which for a matroid $\M=(V,\I)$. We can construct a heuristic based on the stochastic greedy algorithm \cite{mirzasoleiman_etal15} and adapted to partition constraints (Algorithm \ref{alg:ext_partition_stochastic_greedy}): given $\epsilon'>0$, in each round it uniformly samples $\frac{n_j}{b}\log \frac{1}{\epsilon'}$ elements from each part $R_j\sim P_j$, where $n_j:=|P_j|$. And then, it obtains the element with the largest marginal value among elements in $\cup_{j\in[q]}R_j$. 

Even though, we are not able to state any provable guarantee, we use Algorithm \ref{alg:ext_partition_stochastic_greedy} as inner loop for solving the robust problem \eqref{eq:problem_def} with $\ell = \lceil \log \frac{2k}{\epsilon}\rceil$.

\subsection{Pseudo bi-criteria algorithm}\label{sec:final_pseudo}
In this section, we present the pseudo-code of the main algorithm that we use for the experiments in Section \ref{sec:comp-results}. Algorithm \ref{alg:final_robust} works as follows: in a outer loop we obtain an estimate $\gamma$ on the value of the optimal solution $\OPT$ via a binary search. For each guess $\gamma$ we define a set function \(g^\gamma(S):= \frac{1}{k} \sum_{i\in[k]} \min\{f_i(S), \gamma\}\). Then, we run algorithm $\mathcal{A}$ (either \texttt{E-G}, \texttt{E-ThG}, or \texttt{E-StochG}) on $g^\gamma$. If at some point the solution $S$ satisfies $\min_{i\in[k]}f_i(S)\geq (1-\epsilon/2)\gamma$, we stop and update the lower $\lb = \min_{i\in[k]}f_i(S)$, since we find a good candidate. Otherwise, we continue. After finishing round $\tau$, we check if we realize the guarantee $g^\gamma(\cup_{j=1}^\tau S_j) \geq \alpha_\tau\cdot\gamma$. If not, then we stop and update the upper bound $\ub = \gamma$, otherwise we continue. Finally, we stop the binary search whenever $\lb$ and $\ub$ are sufficiently close. We consider factor guarantees (line 17 in Algorithm \ref{alg:final_robust}) $\alpha_\tau = 1-1/2^\tau$ for \texttt{E-G} and or \texttt{E-StochG}, and $\alpha_\tau = 1 - 1/(2-\delta)^\tau$ for \texttt{E-ThG}. 

\begin{algorithm}[h]
\caption{\small Pseudo-code to get bi-criteria solutions}\label{alg:final_robust}
\begin{algorithmic}[1]
\small
\Require $\epsilon>0$, monotone submodular functions $\{f_i\}_{i\in[k]}$, partition matroid $P_1,\ldots,P_q$, and subroutine $\mathcal{A}$: \texttt{E-G}, \texttt{E-ThG}, or \texttt{E-StochG}.
\Ensure sets $S_1,\ldots, S_\ell\in \I$.
\State Compute $\lb$ and $\ub$ as stated above.
\While {$\frac{\ub-\lb}{\ub} > 2\epsilon$}
\State $\gamma = (\ub+\lb)/2$
\For {$\tau = 1,\ldots, \ell$}
\State $S_\tau =\emptyset$.
\State Compute marginals $\rho(e) = g^\gamma(S+e) - g^{\gamma}(S)$ for all $e\in V$.
\If {$\max_e\rho(e) \leq 0$}

\If {$\min_if_i(\cup_{j=1}^\tau S_j)\geq (1-\epsilon)\gamma$}
\State {\bf Update} $\lb = \min_if_i(\cup_{j=1}^\tau S_j)$
\Else
\State {\bf Update} $\ub = \gamma$
\EndIf
\State {\bf Break}

\Else
\State {\bf Obtain} $S_\tau \leftarrow \mathcal{A}(g^\gamma,\cup_{j=1}^{\tau-1}S_j)$

\State 
\If {$g^\gamma(\cup_{j=1}^\tau S_j) < \alpha_\tau\cdot\gamma$}
\State {\bf Update} $\ub = \gamma$.
\State {\bf Break}
\Else
\If {$\min_if_i(\cup_{j=1}^\tau S_j)\geq (1-\epsilon)\gamma$}
\State {\bf Update} $\lb = \min_if_i(\cup_{j=1}^\tau S_j)$
\State {\bf Break}
\Else
\State {\bf Continue}
\EndIf
\EndIf
\EndIf

\EndFor
\EndWhile

\end{algorithmic}
\end{algorithm}
\normalsize

\subsection{Extra experiment - Sensor Placement}\label{sec:extra_experiment}

For this problem we follow the setup in \cite{krause2008robust}. Here, we are given a set of sensors $V$ with fixed locations in a specific region. Each sensor $s$ measures certain phenomena such as temperature, humidity and light, which define a random vector $X_s$. We assume that the set of random variables $X_V$ is distributed according to a multivariate normal distribution, which corresponds to a Gaussian Process (GP). The predictive variance of sensor $s$ after obtaining observations from a subset of sensors $A\subseteq V$ is given by \(\sigma^2_{s|A} = \sigma_s^2 - \Sigma_{sA}\Sigma_{AA}^{-1}\Sigma_{As}\), where $\Sigma_{AA}$ is the covariance matrix of the measurements at the chosen locations $A$, $\Sigma_{sA}$ is the row-vector in $\Sigma$ with row $s$ and columns $A$, and $\sigma_s^2$ is the \emph{a priori} variance of sensor $s$. Traditionally, the goal is to find a subset $A$ that minimizes the predictive variance. However, let us assume that the \emph{a priori} variance $\sigma^2_s$ is constant for all locations $s$ and define the variance reduction $f_s(A) := \Sigma_{sA}\Sigma_{AA}^{-1}\Sigma_{As}$. \cite{das2008algorithms} show that $f_s$ is monotone and submodular for certain distributions. Therefore, minimizing \(\sigma^2_{s|A}\) is equivalent to maximizing $f_s$ when $\sigma_s^2$ is assumed to be constant.

We use the Intel Research Berkeley dataset of $n = 44$ sensors, which contains measurements of temperature (T), humidity (L), and light (L). We consider data of three consecutive days, and we construct the corresponding covariance matrices $\Sigma^T$, $\Sigma^H$, and $\Sigma^L$. 
For our robust criteria, we consider perturbed versions of the average variance reduction for each observation $k=3$, i.e., problem \eqref{eq:problem_def} corresponds to $\max_{A\in \I}\min_{k \in \{T,H,L\}} \{f_k(A)+\sum_{e\in A\cap \Lambda_k}\eta_e\}$, where $f_k(A) = \frac{1}{44}\sum_{s\in[44]}\Sigma^k_{sA}(\Sigma^k_{AA})^{-1}\Sigma^k_{As}$, $\Lambda_k$ is random set of size 15, different in composition for each $k$, and $\eta\sim[0,1]^V$ is an error vector. 

We made 30 random runs considering the number of parts $q = 3$ and budget $b=1$. We report the results of the instances in Figures \ref{fig:experiments_sensors}. We observe that the three tested algorithms clearly outperform \texttt{prevE-G}, either in running time (see box-plot (b) and performance profile (d)) and in the number of function evaluations (see box-plot (a) and performance profile (e)). When we only compared the three tested algorithm, the performance is very similar (see box-plots (a) and (b)), but \texttt{E-StochG} is the most likely to use less number of function calls, see profile (e). In terms of running time \texttt{E-G} and \texttt{E-Stoch} are the most likely to solve the problem faster, see profile (d). Finally, in chart (c) we observe that the stopping rules help to find a good candidate solution earlier by using less elements and at much less cost.

\begin{figure*}[h!]
\begin{center}
\begin{tabular}{ccc}

 \includegraphics[width=0.3\textwidth]{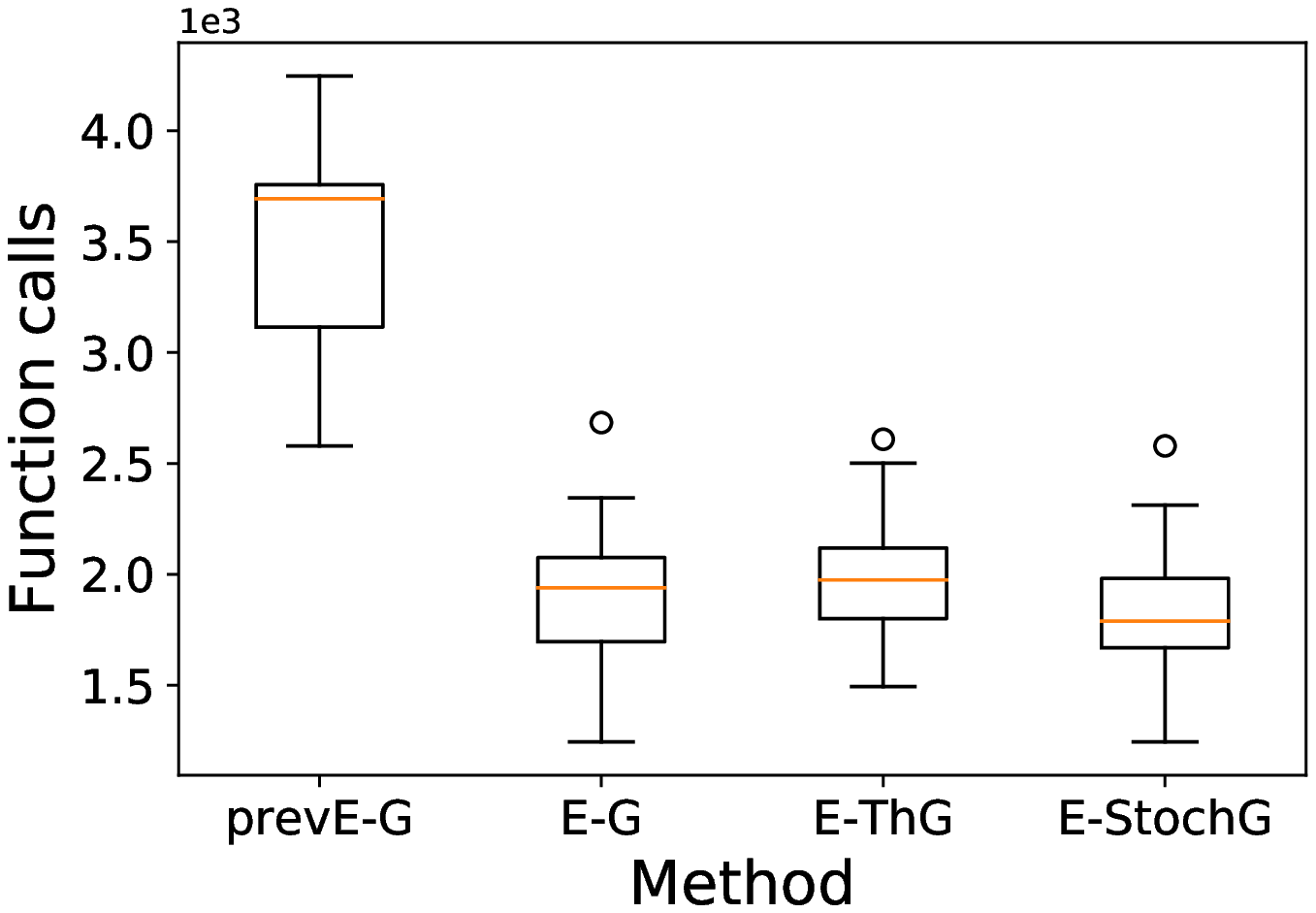} &
\includegraphics[width=0.3\textwidth]{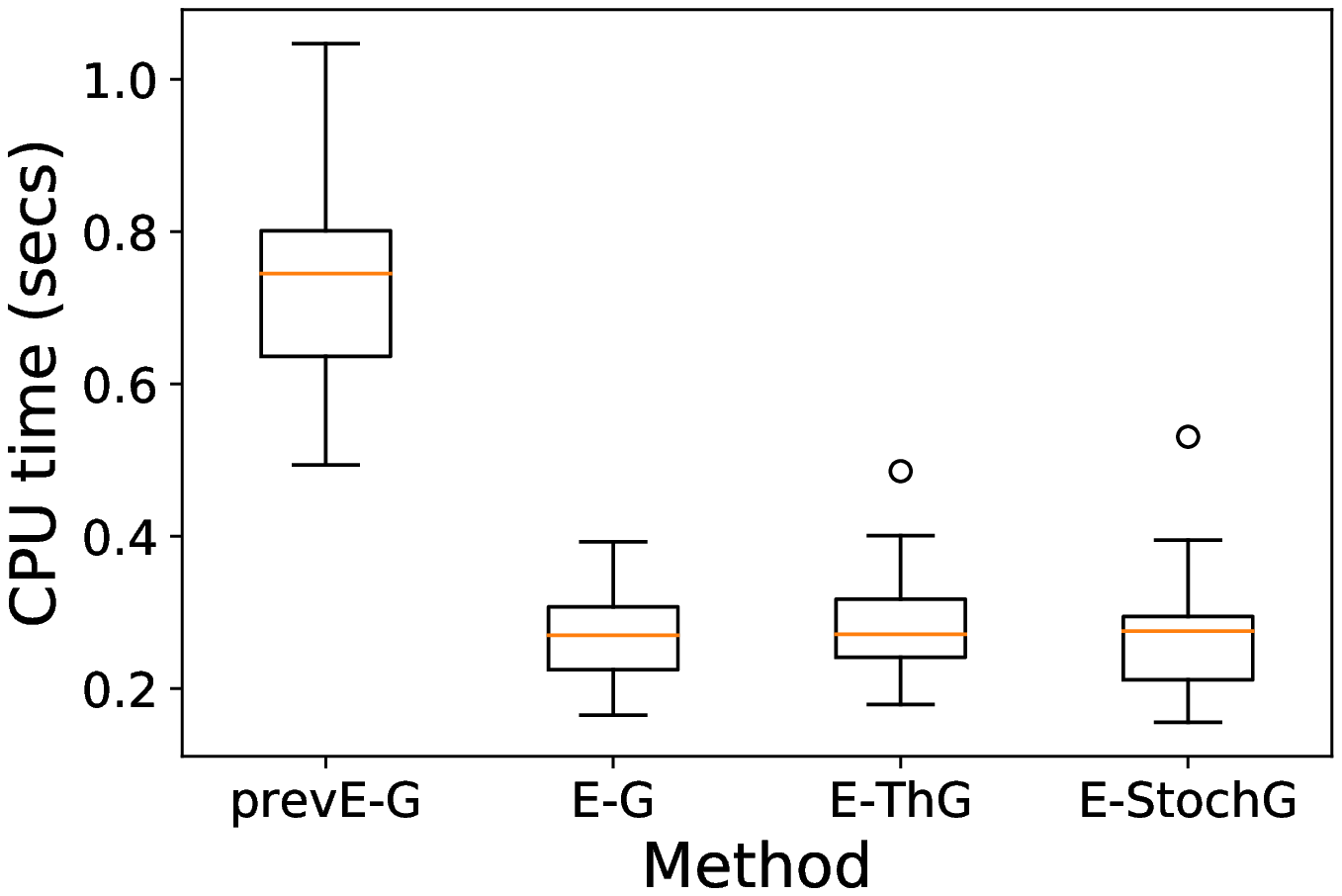} &
\includegraphics[width=0.3\textwidth]{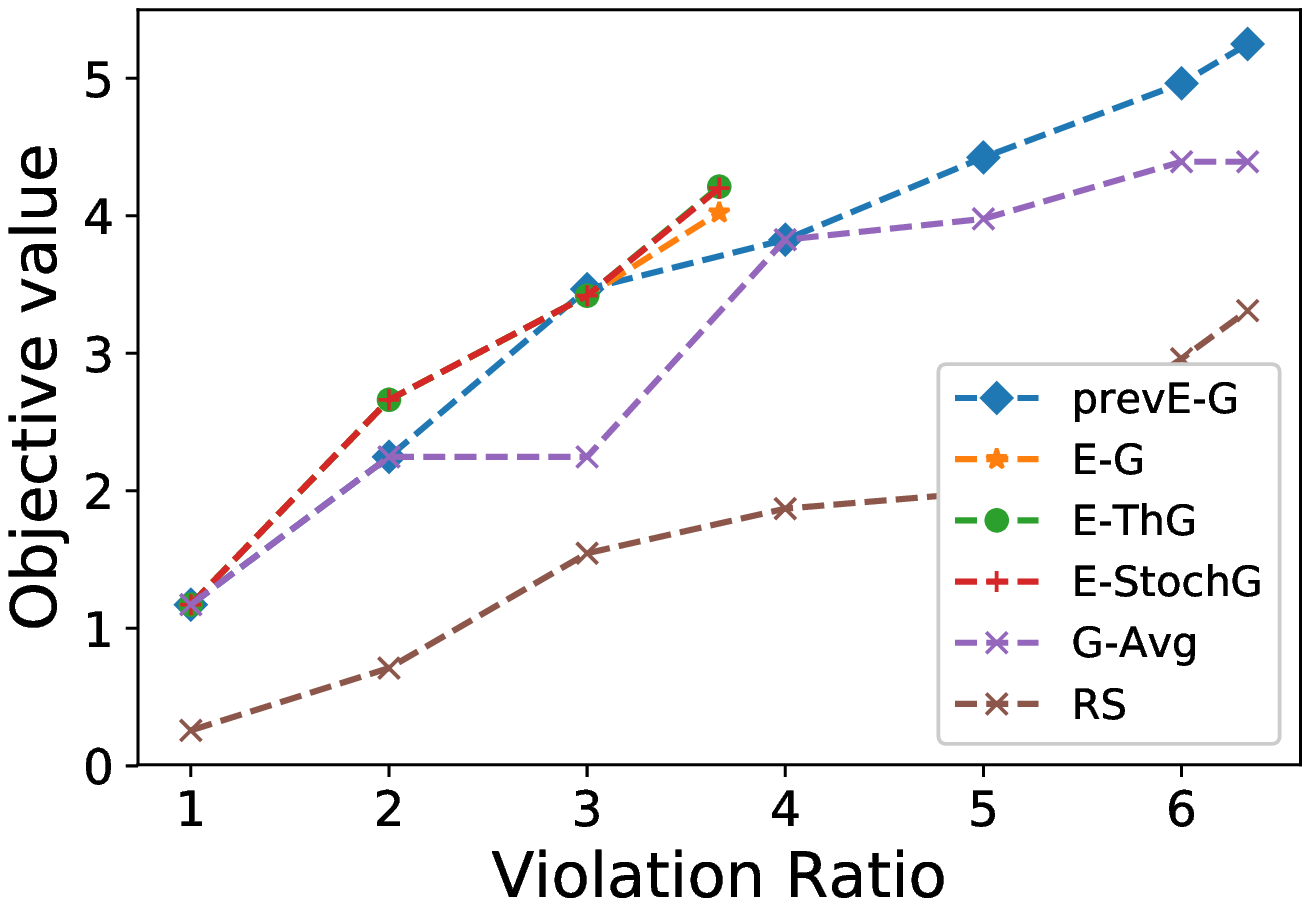} \\
 (a) & (b) & (c) \\
\includegraphics[width=0.3\textwidth]{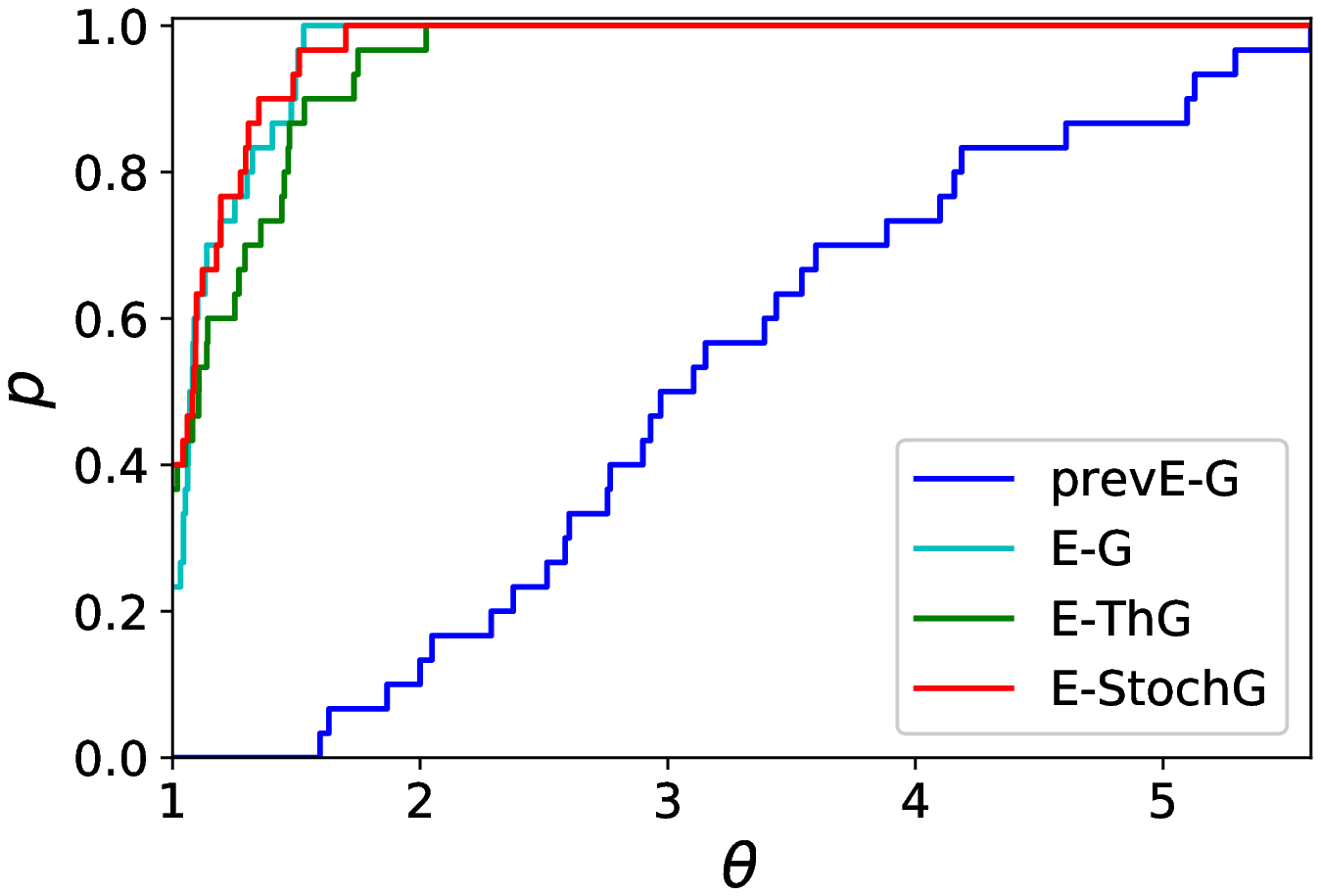} &
\includegraphics[width=0.3\textwidth]{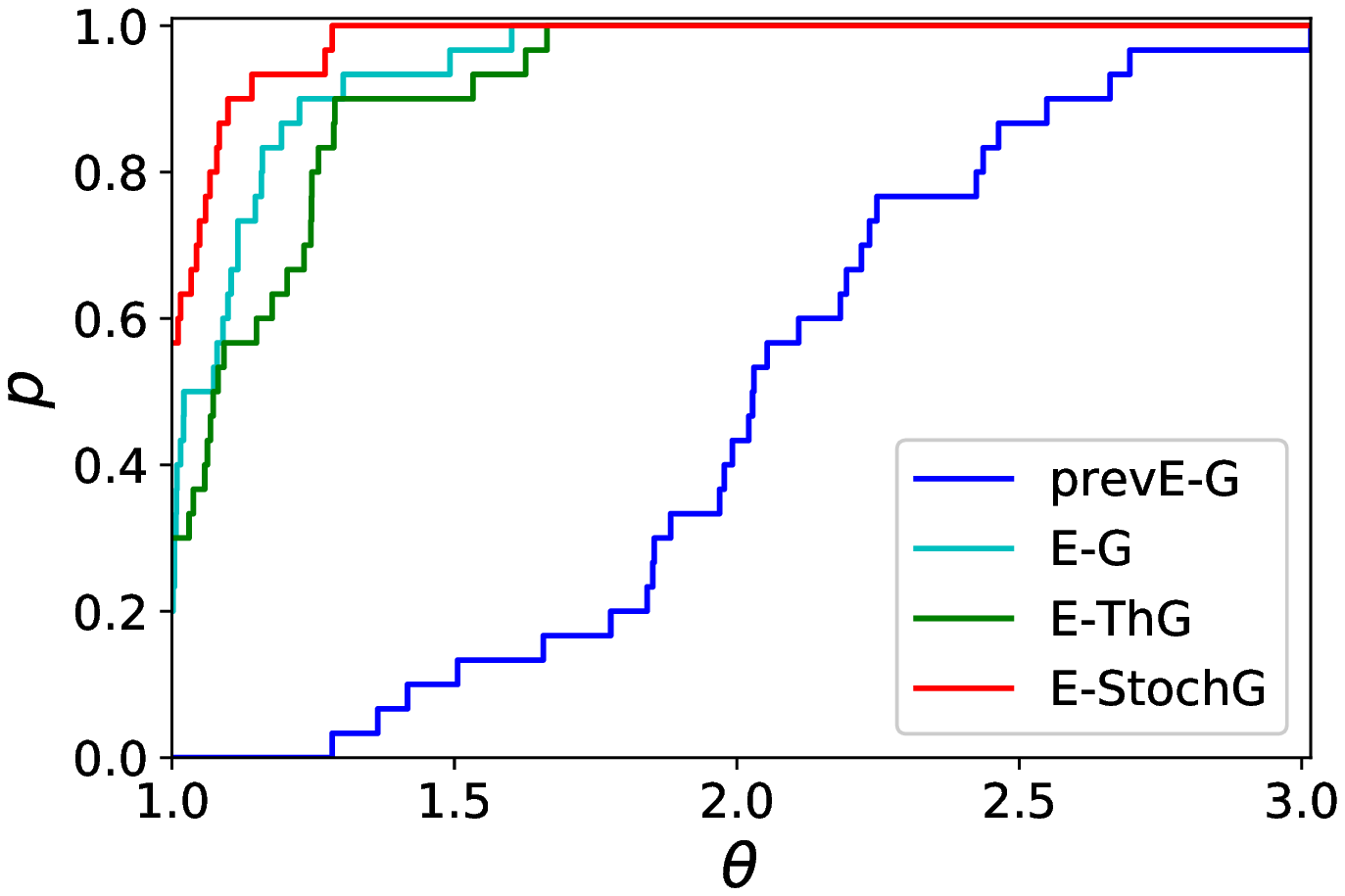} & \\
(d) & (e) & 
 \end{tabular}
 \caption{{\it Sensor Placement}: Box-plots (a) for the function calls and (b) for the running time. In (c) we present the objective value versus the violation ratio in a single run of each method. Performance profiles (d) for the running time and (e) for the function calls.}
\label{fig:experiments_sensors}
\end{center}
\end{figure*}